%% file: main.tex
\documentclass[letterpaper,11pt,dvipsnames]{article}

\usepackage{macros}

\newcommand{\MDP}{\mathsf{MDP}}
\newcommand{\SVP}{\mathsf{SVP}}
\newcommand{\GapETH}{\mathsf{Gap\text{-}ETH}}
\newcommand{\FPT}{\mathsf{FPT}}
\newcommand{\Wone}{\mathsf{W}[1]}

\title{Tight Lower Bound for Approximating Parametrized Maximum Likelihood Decoding under ETH}
\author{Rishav Gupta \\ National University of Singapore \\ \texttt{rishavg@u.nus.edu} \and Bingkai Lin\footnote{State Key Laboratory of Novel Software Technology, Nanjing University, 
Nanjing 210023, China} \\ Nanjing University \\\texttt{lin@nju.edu.cn} \and Xin Zheng \\ Nanjing University\\ \texttt{xinzheng@smail.nju.edu.cn}}
\date{}

\begin{document}

\maketitle
\begin{abstract}
We present a simple deterministic reduction which, assuming the Exponential Time Hypothesis ($\ETH$),  yields tight 
lower bounds for approximating the parameterized Maximum Likelihood Decoding problem ($\mld$) and the parameterized Nearest Codeword Problem ($\ncp$) within some fixed constant factor.
Our starting point is the
$\ETH$-based exponential-time hardness of $\gapmaxlin$ established in \cite{BHIRW24}. We  transform a $\gapmaxlin$
instance into an instance of $\gapkmld$ via a novel
combinatorial object that we call a \emph{cover family}. We provide both a randomized construction
of the required cover families and a subsequent derandomization. Prior to our work, 
$n^{\Omega(k)}$ hardness for 
constant-factor approximation was only shown under the randomized
Gap Exponential Time Hypothesis $\GapETH$ \cite{Man20}, which is a much stronger assumption than $\ETH$. Under $\ETH$,
the strongest known lower bound was $n^{\Omega(k/\operatorname{poly} \log k)}$ due to
\cite{BafnaSM25}.
Unlike previous approaches that rely on reductions from the hardness of approximating $\csp[2]$, our reduction provides a more direct and conceptually simpler route to achieving the optimal lower bounds.
\end{abstract}
\newpage
\input{intro}

\input{prelims_matrix}

\input{maxlin_to_gap_mdp_matrix}
\input{construction}

\section{Acknowledgments}
Supported by the ``111 Cente'' (No. B26023) and NRF grant NRF-NRFI09-0005.
The authors would like to thank the anonymous reviewers for their valuable feedback. Special thanks to Reviewer B for simplifying the proof for \cref{lem:partition_gadget_derandomize_lem}.


\bibliographystyle{alpha}
\bibliography{references}


\end{document}

%% file: intro.tex
\section{Introduction}
The study of error-correcting codes gives rise to various computational problems. In the context of channel coding, one of the most fundamental tasks is to recover the original message from a signal corrupted by noise. In this problem, which is known as the \emph{Maximum Likelihood Decoding} ($\mld$) problem, we are given the parity check matrix $\vec H \in \F_q^{d \times n}$ of a linear code, a vector $\vec{t} \in \F_q^d$, and a parameter $k\in\N$. Our goal is to find a vector $\vec{x} \in \F_q^n$ of Hamming weight at most $k$ such that $\vec H\vec{x} = \vec{t}$. This problem is computationally equivalent to the \emph{Nearest Codeword Problem} ($\ncp$), where we are given the generator matrix $\vec A \in \F_q^{d \times n}$ and the goal is to find a vector $\vec{x} \in \F_q^n$ such that $\|\vec A\vec x-\vec t\|_0\le k$.

The computational hardness of these problems is well-established. They were proven to be NP-hard decades ago~\cite{BMvT78} and are known to be hard to approximate within any constant factor in polynomial time~\cite{ABSS97, DMS03}. Consequently, attention shifted to the framework of \emph{parameterized complexity}~\cite{DF99}, where the runtime is analyzed with respect to the parameter $k$.

In the parameterized setting, the hardness of the {exact} problem was settled at first. Downey and Fellows~\cite{downey1999parametrized} proved that $\kmld$ parameterized by $k$ is $\Wone$-hard, ruling out exact FPT algorithms, i.e., algorithms that run in $f(k) \cdot n^{O(1)}$ time, under $\Wone\neq\FPT$. However, the question of whether \emph{approximation} could yield fixed parameter tractability remained a major open problem for over two decades.

This question was finally resolved in a series of breakthrough results. The authors of \cite{BBE+21} established that $\kncp$ is $\Wone$-hard to approximate within any constant factor.
They also proved the parameterized inapproximability of the \emph{Minimum Distance Problem} ($\MDP$) over $\F_2$ and the \emph{Shortest Vector Problem} ($\SVP$) in $\ell_p$ norm for every $p>1$.
Following this, Bennett, Cheraghchi, Guruswami, and Ribeiro~\cite{BCGR23}  generalized these results, proving the inapproximability of $\kmdp$ over all finite fields and $\ksvp$ in general $\ell_p$ norms.
These works painted a comprehensive picture of the FPT inapproximability of coding problems. With the W[1]-hardness established, the focus has shifted to the fine-grained complexity: 
\begin{center}
    \textit{What is the precise running time required to approximate $\kmld$ when $k$ is small?}
\end{center}

Under the standard gap-free hypothesis, the \emph{Exponential Time Hypothesis} ($\ETH$), \cite{BBE+21} ruled out $n^{O((\log k)^{1/{2.01}})}$-time algorithms for constant gap $\kncp$. Their reduction incurred an exponential blow-up in the parameter, which results in the relatively weak lower bound.
Very recently, Li, Lin, and Liu~\cite{LLL24} developed a new gap-creating reduction with polynomial parameter growth, hence improving the lower bound under standard $\ETH$, ruling out $n^{o(k^{1/3})}$-time algorithms that solve $\gapkncp$ for any $1<\gamma<\frac{3}{2}$.
Through a self-reduction, they further ruled out $n^{o(k^\varepsilon)}$-time algorithms where $\varepsilon=\frac{1}{\text{poly}\log\gamma}$ for any constant gap $\gamma>1$.
While this was a major step forward, a substantial gap remains between the lower bound and the trivial $n^{O(k)}$ upper bound.

To date, the only tight lower bound comes from the work of Manurangsi~\cite{Man20}, which showed that no  approximation algorithm exists for these problems in time $n^{o(k)}$ for any constant factor. However, this result relies on the randomized \emph{Gap-Exponential Time Hypothesis} ($\GapETH$), which is a significantly stronger assumption than $\ETH$.

The above-mentioned results establish hardness for \emph{all} constant approximation factors. However, even obtaining an optimal lower bound under $\ETH$ for \emph{some} constant approximation factor remained open.
Prior to our work, the strongest known $\ETH$-based hardness for \emph{some} approximation factor was established via a chain of reductions: first, by invoking the hardness of approximating $k$-variable $\csp[2]$~\cite{optethpih,BafnaSM25}, then reducing $k$-variable $\csp[2]$ to $\kexactcover$~\cite{grs24}, and finally reducing $\kexactcover$ to $\gapkmldq$  \cite{ABSS97}.
This approach yields suboptimal running-time lower bounds: namely $n^{\,k/2^{\Omega(\sqrt{\log k}\cdot \log\log k)}}$ via~\cite{optethpih} and $n^{\,k/\log^{C}(k)}$ via~\cite{BafnaSM25}.

\paragraph{Our Contribution.} In this work, we prove that the tight $n^{\Omega(k)}$ lower bound holds from some constant approximation factor under the standard non-deterministic Exponential Time Hypothesis, removing the need for the stronger Gap-ETH assumption.

\begin{theorem}[Main Result]
\label{thm: mld_main}
    Assuming $\ETH$, there exist constants $\gamma > 1$ and $\delta>0$, such that for any algorithm that takes as input a matrix $\vec{H}\in \F_q^{d \times N}$, a vector $\vec{u}\in\F_q^{d}$ and a parameter $k\in\N$, it must take $ N^{\delta k}$ time to distinguish the following two cases:
    \begin{itemize}
        \item $\yes$ Case: There exists $\vec x\in\F_q^{N}$ with $\|\vec x\|_0\le k$ such that $\vec H\vec x=\vec u$.
        \item $\no$ Case: For every $\vec x\in\F_q^{N}$ with $\|\vec x\|_0\le \gamma k$, we have $\vec H\vec x\neq\vec u$.
    \end{itemize}
\end{theorem}

If we allow  $k$ to be some function of $N$, the proof of \cref{thm: mld_main} actually gives an $ N^{\delta k}$-time lower bound for $\gapkmldq$ when $k\le O((\log\log N)^{0.49})$. Assuming randomized $\ETH$, we can obtain the same lower bound for $k\le N^\varepsilon$ for some constant $\varepsilon>0$.

\begin{theorem}
\label{thm: mld_main_2}
    Assuming randomized $\ETH$, there exist constants $\gamma > 1$ and $\delta,\varepsilon>0$, such that for any algorithm that takes as input a matrix $\vec{H}\in \F_q^{d \times N}$, a vector $\vec{u}\in\F_q^{d}$ and an integer $k\in\N$ such that $2\le k\le N^\varepsilon$, it must take $N^{\delta k}$ time to distinguish the following two cases:
    \begin{itemize}
        \item $\yes$ Case: There exists $\vec x\in\F_q^{N}$ with $\|\vec x\|_0\le k$ such that $\vec H\vec x=\vec u$.
        \item $\no$ Case: For every $\vec x\in\F_q^{N}$ with $\|\vec x\|_0\le \gamma k$, we have $\vec H\vec x\neq\vec u$.
    \end{itemize}
\end{theorem}

\paragraph{Our Techniques.}
Our hardness proof builds upon the $\ETH$-hardness of $\gapmaxlinq$ established by~\cite{BHIRW24}. In this problem, one must distinguish
whether a system of linear equations is almost satisfiable (a $c$-fraction of equations hold), or is far from satisfiable (at most an $s$-fraction hold).
Specifically, the result of~\cite{BHIRW24} shows that assuming $\ETH$, there is no algorithm that can solve $\gapmaxlinq$ in time $2^{o(n)}$.
By considering the dual of a hard $\gapmaxlinq$ instance, we first obtain a hard instance of $\gapmldq$, where the goal is to determine whether a target vector can be expressed as a linear combination of at most $\ell$ columns from a matrix, or it requires at least $\gamma \ell$ columns. This duality transformation implies that the non-parameterized version of $\gapkmldq$ also cannot be solved in time $2^{o(n)}$.

The core of our contribution is a reduction from this non-parameterized $\gapmldq$ to its parameterized version $\gapkmldq$. 
We let each vector in the output $\kmld$ instance be the sum of $\ell/k$ vectors in the original $\mld$ instance. Thus, a solution of size $\ell$ in the $\mld$ instance can be represented by only $k$ vectors in the $\kmld$ instance. However, the naive brute-force enumeration of all combinations of size $\ell/k$ produces an instance of size $N = \binom{n}{\ell/k}$, hence the reduction yields only an $N^{\Omega(k/\log k)}$ lower bound for $\gapkmldq$.

We achieve a tight lower bound via a novel combinatorial structure we term \emph{cover families}. 
A cover family is a sufficiently small collection of subsets $\mathcal{S}$ over a universe $\mathcal{U}$, such that any \emph{small} subset of $\mathcal{U}$ can be expressed as the union of $k$ pairwise disjoint sets from $\mathcal{S}$, whereas no \emph{large} subset of $\mathcal{U}$ can be expressed as the union of even $\gamma k$ sets from $\mathcal{S}$. To construct suitable cover families, we introduce an intermediate combinatorial object called \emph{balanced partition families}. A balanced partition family is a collection of partitions $\mathcal{P}$ of a universe $\mathcal{U}$ such that every partition in the family is roughly balanced, and every subset of $\mathcal{U}$ of an appropriate size is (almost) equipartitioned by some partition $P \in \mathcal{P}$. 
We first show that independently sampling random partitions suffices to obtain the desired balanced partition family. We then give a deterministic construction of such balanced partition families using hypercube set systems. Finally we  use a standard reduction from $\gapkmldq$ to $\gapkncpq$ to obtain the hardness of the latter.


\paragraph{Future Directions.}Assuming $\ETH$, we obtain optimal lower bounds for $\gapkmldq$ and $\gapkncpq$ for some constant approximation factor $\gamma>1$. A natural question is whether this lower bound can be extended to an arbitrary constant factor assuming only $\ETH$. In contrast, under $\GapETH$ the corresponding conclusion is known to hold for every constant $\gamma>1$ via \cite{Man20}. Establishing such a result under $\ETH$ would close the current knowledge gap between $\ETH$ and $\GapETH$ through the lens of parameterized coding-theoretic problems. 
We also wonder whether $n^{\Omega(k)}$ lower bound for $\gapkmdp$ or other related problems can be established under $\ETH$.
It is worth mentioning that \cite{BafnaSM25} rules out $n^{k/\log^C k}$-time algorithms that approximates $\kncp$ and $k\textsf{-NVP}$ to any constant factor under $\ETH$.
Another open direction is to understand what further consequences can be derived from the $\ETH$-hardness of $\gapmaxlinq$ established in \cite{BHIRW24}. Beyond our work, the only other result we are aware of that leverages this hardness result is \cite{agm25}. Lastly, one can also explore other applications of the combinatorial objects, \emph{cover family} and \emph{balanced partition family} which are introduced in this paper.

\paragraph{Paper Organization.} In \cref{sec:prelim_m}  we define necessary notation and introduce useful tools from
the literature. First in \cref{sec:maxlin_to_mld} we present a reduction from $\gapmaxlinq$ to $\gapkmldq$ using \textit{cover family}. Then in \cref{sec:constr} we give a randomized as well as a deterministic construction of a suitable \textit{cover family} using an intermediate combinatorial object, \textit{balanced partition family}. 



%% file: prelims_matrix.tex
\section{Preliminaries}
\label{sec:prelim_m}
We begin by formally defining the computational problems that will be studied throughout the paper. For each problem, we specify its input, the underlying computational goal, and any associated promise conditions.
\subsection{Computational Problems}
We first define the standard $\sat[k]$ problem. For an integer $k \geq 2$, a $\sat[k]$ formula over $n$ boolean variables is the conjunction of clauses, where each clause is the disjunction of $k$ literals. That is, $\sat[k]$ formulas have the form $\bigwedge_{i=1}^m \bigvee_{j=1}^k b_{i,j}$, where $b_{i,j} = x_k$ or $b_{i,j} = \neg x_k$ for some boolean variable $x_k$.

\begin{definition} [{$\sat[k]$}]
	For any $k \geq 2$, the decision problem $\sat[k]$ is defined as follows. The input is a $\sat[k]$ formula. It is a \yes instance if there exists an assignment to the variables that makes the formula evaluate to true and a \no instance otherwise.
\end{definition}
We write $\sat[k]_C$ for a $\sat[k]$ instance where each variable $x_i$ is contained in at most $C$ clauses. We also define the corresponding optimization version Max-$\sat[k]$.

\begin{definition}[{Max-$\sat[k]$}]
	For any $k \geq 2$, the decision problem Max-$\sat[k]$ is defined as follows. The input is a $\sat[k]$ formula and an integer $S \geq 1$. It is a \yes instance if there exists an assignment to the variables such that at least $S$ of the clauses evaluate to true and a \no instance otherwise.
\end{definition}

\begin{definition}[$\gapmaxlinq$]
    A $\gapmaxlinq$ instance $(\vec A,\vec b)$ consists of an $m \times n$ matrix $\vec A \in \F_q^{m \times n}$ and a vector $\vec{b} \in \F_q^m$.
    The objective of the problem is to distinguish between the following cases.
    \begin{itemize}
        \item $\yes$ Case: There exists $\vec{x}\in \F_q^n$ such that $\|\vec A\vec x-\vec b\|_0\le (1-c)m$.
        \item $\no$ Case: For every $\vec{x}\in \F_q^n$, $\|\vec A\vec x-\vec b\|_0> (1-s)m$.
    \end{itemize}  
\end{definition}

We need to mention that this problem is equivalent to the non-parameterized $\gapmldq[\frac{1-s}{1-c}]$.
Now we define the dual problem of $\maxlin$, Maximum Likelihood Decoding Problem.

\begin{definition}[$\gapmldq$]\label{def:gmld}
    A $\gapmldq$ instance $(\vec H,\vec u,\ell)$ consists of a $d\times n$ matrix $\vec H \in \F_q^{d\times n}$, a target vector $\vec u \in \F_q^d$ and a value $\ell$. The goal is to distinguish between the following cases.
    \begin{itemize}
        \item $\yes$ Case: There exists $\vec x\in\F_q^{n}$ with $\|\vec x\|_0\le \ell$ such that $\vec H\vec x=\vec u$.
        \item $\no$ Case: For every $\vec x\in\F_q^{n}$ with $\|\vec x\|_0\le \gamma\ell$, we have $\vec H\vec x\neq\vec u$.
    \end{itemize}
\end{definition}

We now define the parametrized version of $\gapmldq$, which is almost the same as its non-parameterized version, just using $k:=\ell$ as the parameter.
\begin{definition}[$\gapkmldq$]\label{def:gap-MLD}\label{def:gapkmld}
    A $\gapkmldq$ instance $(\vec H,\vec u)$ consists of a $d\times n$ matrix $\vec H \in \F_q^{d\times n}$, a target vector $\vec u \in \F_q^d$ and a parameter $k\in\mathbb{N}$. The goal is to distinguish between the following cases.
    \begin{itemize}
        \item $\yes$ Case: There exists $\vec x\in\F_q^{n}$ with $\|\vec x\|_0\le k$ such that $\vec H\vec x=\vec u$.
        \item $\no$ Case: For every $\vec x\in\F_q^{n}$ with $\|\vec x\|_0\le \gamma k$, we have $\vec H\vec x\neq\vec u$.
    \end{itemize}
\end{definition}

Next we turn to the parameterized Nearest Codeword Problem, which is closely related to MLD.

\begin{definition}[$\gapkncpq$]
\label{def:gap-NCP}
    A $\gapkncpq$ instance $(\vec A,\vec t)$ consists of a $n\times d$ matrix $\vec A \in \F_q^{n\times d}$, a target vector $\vec t \in \F_q^n$ and a parameter $k\in\mathbb{N}$. The goal is to distinguish between the following cases.
    \begin{itemize}
        \item $\yes$ Case: There exists $\vec{x}\in \F_q^d$ such that $\|\vec A\vec x-\vec t\|_0\le k$.
        \item $\no$ Case: For every $\vec{x}\in \F_q^d$, $\|\vec A\vec x-\vec t\|_0> \gamma k$.
    \end{itemize}

        
\end{definition}

Similarly we define parametrized versions of Closest Vector Problem ($\CVP$). 
For $\vec x \in \Z^d$ we will use the notation   $\| \vec{x} \|_p = \left( \sum_{i=1}^n |x_i|^p \right)^{1/p}$, in the following definition.

\begin{definition}[$\gapkcvpp$]
\label{def:gap-CVP}
    A $\gapkcvpp$ instance $(\vec A,\vec t)$ consists of a $n\times d$ matrix $\vec A \in \Z^{n\times d}$, a target vector $\vec t \in \Z^n$ and a parameter $k\in\mathbb{N}$. The goal is to distinguish between the following cases.
    \begin{itemize}
        \item $\yes$ Case: There exists $\vec{x}\in \Z^d$ such that $\|\vec A\vec x-\vec t\|_p\le k$.
        \item $\no$ Case: For every $\vec{x}\in \Z^d$, $\|\vec A\vec x-\vec t\|_p> \gamma k$.
    \end{itemize}

        
\end{definition}

\subsection{Fine-grained Hardness Assumptions}
We introduce the following fine-grained hardness assumptions, Exponential time hypothesis $\ETH$ and its corresponding gap version $\geth$ below.
\begin{definition}[Exponential Time Hypothesis ($\ETH$), \cite{IP01}]
	There exists $\delta >0$ such that any  algorithm which solves $\sat$ must take $2^{\delta n}$ time.
\end{definition}

\begin{definition}[$\geth$, \cite{gapeth}]
    There exists $\delta>0$ and $0<\eta<1$ such that given a $\sat$ instance with $n$ variables and $m$ clauses, any  algorithm which can distinguish between the cases if all $m$ clauses are satisfiable and one in which no assignment satisfies more than $\eta$-fraction of the clauses, must take $2^{\delta n}$ time.
\end{definition}

We now state the Sparsification Lemma.
\begin{lemma}[Sparsification Lemma, \cite{IP01}] \label{lemma:sp}
    Let $\varepsilon > 0$, $k \geq 3$ be constants. There is a $2^{\varepsilon n} \cdot \text{poly}(n)$ time algorithm that takes a $k$-CNF $F$ on $n$ variables and produces $F_1, \dots, F_{2^{\varepsilon n}}$, $2^{\varepsilon n}$ $k$-CNFs such that $F$ is satisfied if and only if $\bigvee_i F_i$ is satisfied and each $F_i$ has $n$ variables and $n \cdot \big(\frac{k}{\varepsilon}\big)^{O(k)}$ clauses. In fact, each variable is in at most $\mathrm{poly}\big(\frac{1}{\varepsilon}\big)$ clauses, and the $F_i$ are over the same variables as $F$.
\end{lemma}

We will now state the following result from \cite{BHIRW24} which gives us a very strong starting point for establishing hardness results.

\begin{theorem}[\cite{BHIRW24}, Theorem 6.3] \label{thm:satToLinRed}
    For every finite field $\F_q$, there exists constants $c$ and $s$ where $0<s<c<1$, such that the following holds. There exists a polynomial time reduction which takes a $\sat_C$ instance with $n$ variables, and outputs a $\gapmaxlinq$ with $n'=O(n)$ variables and $m'=O(n)$ equations.
\end{theorem}

The sparsification \cref{lemma:sp} and Tovey's reduction \cite{TOVEY} together tell us that if $\ETH$ is true, then $\sat_4$ over $n$ variables must take $2^{\delta n}$ time for some $\delta>0$.
Together with \Cref{thm:satToLinRed}, we find that if $\ETH$ holds, then for some $0<s<c<1$ and $C>0$, any algorithm which solves $\gapmaxlinq$ with $n$ variables and $m\leq Cn$ clauses, must take $2^{\delta n}$ time for some $\delta > 0$. We state it as the following corollary.

\begin{corollary}[$\gapmaxlinq$ is $\ETH$-Hard] \label{cor:linEthHard}
	For every finite field $\F_q$, there exists constants $0<s<c<1$ and $C > 0$ such that unless $\mathsf{ETH}$ is false, any algorithm for $\gapmaxlinq$ with $n$ variables and $m\leq Cn$ equations must take $2^{\delta n}$ time for some $\delta > 0$.
\end{corollary}

%% file: maxlin_to_gap_mdp_matrix.tex
\section{Reduction from \texorpdfstring{$\gapmaxlin$ to $\gapkmldq$}{MAXLIN to k-MLD}}
\label{sec:maxlin_to_mld}

In this section, we present reductions from $\gapmaxlin$ to $\gapkmldq$.  We first give a simple reduction that yields $N^{\Omega(k/\log k)}$ lower bound. Then we provide a reduction  for the tight lower bound. The reduction contains two steps.

\paragraph{Step 1: Reduction from $\gapmaxlin$ to $\gapmldq$.}
\begin{lemma}\label{lem:maxlin2mld}
    There exists a polynomial-time reduction that takes as input a $\gapmaxlin$ instance $(\vec A,\vec b)$ where $\vec{A} \in \mathbb{F}_q^{m \times n}$, $\vec{b} \in \mathbb{F}_q^m$ and $0<s<c<1$. The reduction outputs a $\gapmldq$ instance $(\vec H,\vec u,\ell)$ where $\vec{H}\in \mathbb{F}_q^{d\times m}\ (d\le m)$, $\vec{u} \in \mathbb{F}_q^d$, $\ell = (1-c)m$ and $\gamma = \frac{1-s}{1-c}$. This reduction satisfies the following properties:
    \begin{itemize}
        \item \textbf{Completeness:} If there exist $\vec{x} \in \mathbb{F}_q^n$ and $\vec{e} \in \mathbb{F}_q^m$ such that $\|\vec{e}\|_0 \le (1-c)m$ and $\vec{A}\vec{x} + \vec{e} = \vec{b}$, then there exists $\vec x'\in\F_q^m$ with $\|\vec x'\|_0\le\ell$ such that $\vec H\vec x'=\vec{u}$.
        
        \item \textbf{Soundness:} If for every $\vec{x} \in \mathbb{F}_q^n$ and $\vec{e} \in \mathbb{F}_q^m$ satisfying $\vec{A}\vec{x} + \vec{e} = \vec{b}$, we have $\|\vec{e}\|_0 > (1-s)m$, then there is no $\vec x'\in\F_q^m$ with $\|\vec x'\|_0\le\gamma\ell$ such that $\vec H\vec x'=\vec{u}$.
    \end{itemize}
\end{lemma}
\begin{proof}
    Let $\vec H$ denote any parity-check matrix for the code generated by $\vec A$, thus
        $\vec H \in \F_q^{d \times m}$ and $\vec H\vec A=\vec 0$ where $d:=m - \operatorname{rank}(\vec A)$.
    Given $\vec A$, the matrix $\vec H$ can be computed in polynomial time via Gram–Schmidt orthogonalization applied to an appropriate basis of $\F_q^m$. We now convert the $\gapmaxlin$ instance $(\vec A, \vec b)$ to its dual form by multiplying both sides of the equation
    $\vec A \vec x + \vec e = \vec b$ by $\vec H$ on the left.  
    Setting $\vec u := \vec H \vec b$, $\gamma:=\frac{1-s}{1-c}$ and $\ell:=(1-c)m$,
    we obtain the corresponding dual instance $(\vec H, \vec u)$ satisfying:
    \begin{itemize}
        \item \textbf{Completeness.}
        If there exists $\vec x\in\F_q^n$ and $\vec e\in\F_q^m$ such that $\|\vec e\|_0\le\ell=(1-c)m$ and $\vec A\vec x+\vec e=\vec b$, then $\vec H\vec A\vec x+\vec H\vec e=\vec H\vec b$, hence $\vec H\vec e=\vec u$.
    
        \item \textbf{Soundness.}
        Assume for the sake of contradiction that there exists $\vec e\in\F_q^m$ with $\|\vec e\|_0\le\gamma\ell=(1-s)m$ such that $\vec H\vec e=\vec u$, then we have $\vec H(\vec e-\vec b)=\vec 0$, hence $\vec e-\vec b=\vec A\vec x$ for some $\vec x\in\F_q^n$.
    \end{itemize}
    Hence we get a polynomial-time reduction from $\gapmaxlin$ to $\gapmldq$.
\end{proof}

\paragraph{Step 2: Reduction from $\gapmldq$ to $\gapkmldq$.}
Next we will show a reduction from a non-parameterized $\mld$ instance to a parameterized one.
Consider a $\gapmldq$ instance $(\vec M,\vec u,\ell)$.
To obtain the parametrized instance, we need to ``scale down’’ the solution size from $\ell$ to $k$
by grouping the column vectors of $\vec{M}$ into appropriately sized blocks, and construct a new matrix in which each column vector behaves like an aggregated vector.  
This grouping ensures that, in the $\yes$ case, selecting $k$ such aggregated vectors corresponds to selecting $\ell$ original vectors. For the reduction to be sound in the $\no$ case, we additionally require that the size of each group be bounded by $\ell/k$, so that any choice of at most $\gamma k$ aggregated vectors corresponds to at most $\gamma \ell$ original column vectors in $\vec{M}$.

\begin{lemma}\label{lem:naive_group}
    For every $\varepsilon>0$, there exists a reduction that takes as input an integer $k\in\N$, a $\gapmldq$ instance $(\vec M,\vec u,\ell)$ where $\vec M\in \mathbb{F}_q^{d\times m}$, $\vec{u} \in \mathbb{F}_q^d$, and $\frac{k}{\varepsilon}<\ell<\frac{m}{\gamma}$. The reduction outputs a new matrix $\vec{M}_k\in\F_q^{d\times m'}$ in $(m')^{O(1)}$ time, which satisfies the following properties:
    \begin{itemize}
        \item \textbf{Size:} $m'= 2^{O(\frac{m \log k}{k})}$.
        
        \item \textbf{Completeness:}
        If there exists $\vec{x}\in\F_q^{m}$ with $\|\vec x\|_0\le\ell$ such that $\vec M\vec x=\vec{u}$, then there exists $\vec y\in\F_q^{m'}$ with $\|\vec y\|_0\le k$ such that $\vec M_k\vec y=\vec{u}$.
        
        \item \textbf{Soundness:}
        If there is no $\vec{x}\in\F_q^{m}$ with $\|\vec x\|_0\le\gamma\ell$ such that $\vec M\vec x=\vec{u}$, then for $\gamma' = \gamma-\varepsilon$, there is no $\vec{y}\in\F_q^{m'}$ with $\|\vec y\|_0\le\gamma' k$ such that $\vec M_k\vec y=\vec{u}$.
    \end{itemize}
\end{lemma}
\begin{proof}
    Let $r:=\big\lceil \frac{\ell}{k} \big\rceil$, $A:=\{\vec\alpha\in\F_q^m:\|\vec \alpha\|_0\le r\}$, and $m':=|A|$.
    We construct the new matrix $\vec M_k\in\F_q^{d\times m'}$ as follows.
    We establish a bijection between the column indices of $\vec M_k$ and $A$, and use $\vec\alpha\in A$ to denote a column index of $\vec M_k$.
    For each $\vec\alpha\in A$, we let the $\vec\alpha$-th column of $\vec M_k$ be
    $$
        \vec M_k[\vec\alpha] = \vec M\vec \alpha \in\F_q^d.
    $$
    Each column vector in $\vec{M}_k$ is therefore a linear combination of at most $r$ column vectors from $\vec M$.
    Using the fact that $r= \big\lceil \frac{\ell}{k}\big\rceil <\frac{m}{k}$, the size of $\vec M_k$ satisfies
    $$
        m' = \sum_{i=1}^r (q-1)^i\binom{m}{i} < rq^{\frac{m}{k}} \binom{m}{\frac{m}{k}}
        < rq^\frac{m}{k}(\e k)^{\frac{m}{k}}
        = 2^{O(\frac{m  \log k}{k})}.
    $$
    
    \subparagraph{Completeness.}
    Assume that there exists $\vec{x}\in\F_q^{m}$ with $\|\vec x\|_0\le\ell$ such that $\vec M\vec x=\vec{u}$.
    We construct $k$ vectors $\vec x_1,\dots,\vec x_k$ such that $\|\vec x_i\|_0\le r$ and $\sum_{i=1}^k\vec x_i=\vec x$. Such vectors exist and can be constructed in polynomial time.
    Let $\vec y\in\F_q^{m'}$ be such that $\vec y[\vec x_i]=1$ for every $i\in[k]$, and $\vec y[\vec\alpha]=0$ for any other $\vec \alpha\in A$.
    By construction, we have
    $$
        \vec M_k\vec y
        =\sum_{i=1}^k \vec M_k[\vec x_i]
        =\sum_{i=1}^k \vec M\vec x_i
        =\vec M\vec x=\vec u.
    $$
        
    \subparagraph{Soundness.}
    Assume for the sake of contradiction that there exists $\vec{y}\in\F_q^{m'}$ with $\|\vec y\|_0\le\gamma' k$ such that $\vec M_k\vec y=\vec{u}$.
    Let $X:=\supp(\vec y)$, then $|X|=\|\vec y\|_0\le \gamma' k$, and we have
    $$
        \vec u=\vec M_k\vec y=\sum_{\vec \alpha\in X}\vec M_k[\vec\alpha]\cdot \vec y[\vec\alpha] = \sum_{\vec\alpha\in X}\vec M\vec\alpha\cdot \vec y[\vec\alpha].
    $$
    Let $\vec x:=\sum_{\vec\alpha\in X}\vec\alpha\cdot \vec y[\vec\alpha]$, then we have $\vec M\vec x=\vec u$. Moreover, since $\|\vec \alpha\|_0\le r$ for every $\vec\alpha\in A$, we have
    $$
        \|\vec x\|_0
        \le \sum_{\vec\alpha\in X}\|\vec\alpha\|_0
        \le r|X|
        = \gamma'k\ceil{\frac{\ell}{k}}
        \le \gamma'(\ell+k).
    $$
    Thus when $\ell>\frac{k}{\varepsilon}$, we have $\|\vec x\|_0<(\gamma'+\varepsilon)\ell=\gamma\ell$, proving the soundness.
\end{proof}

By \cref{lem:maxlin2mld} and \cref{lem:naive_group}, we obtain a reduction from $\gapmaxlin$ with $m$ equations to $\gapkmldq$ with size $N\le 2^{O(\frac{m \log k}{k})}$ for some constant $\gamma > 1$, hence ruling out $N^{o(\frac{k}{\log k})}$-time algorithms for solving $\gapkmldq$.

Note that in our construction, when we group the vectors to form aggregated vectors, we chose all subsets of size at most $\ceil{\ell/k}$. This brute approach already gave us an almost tight lower bound for $\gapkmldq$. To get a tight lower bound we would want to get a reduction where the number of aggregated vectors is upper bounded by $2^{O(m/k)}$. In the following section, we will make use of a combinatorial gadget, called the ``cover family''. This gadget will essentially enable us to do the grouping more cleverly and getting a tight lower bound.

\subsection{Grouping using Cover Families}
To get a tighter bound in the above reduction, we consider reducing the number of columns in $\vec M_k$. This requires us to find a smaller subset $A\subseteq \F_q^m$ that maintains the completeness and soundness requirements as follows:
\begin{itemize}
    \item (Completeness) For every $\vec x\in\F_q^m$ with $\|\vec x\|_0\le \ell$, there exist $\vec x_1,\dots,\vec x_k\in A$ such that $\sum_{i=1}^k\vec x_i=\vec x$.
    \item (Soundness) For every $\vec x\in\F_q^m$, if there exist $\vec x_1,\dots,\vec x_{\gamma' k}\in A$ such that $\sum_{i=1}^{\gamma'k}\vec x_i=\vec x$, then $\|\vec x\|_0\le \gamma\ell$.
\end{itemize}
To construct such set $A$, we will make use of a combinatorial gadget which we call ``cover family''.
It produces a gap between the number of sets required to cover a small set and a large set.
Let $A$ be the set of all vectors whose support is in the cover family, then $A$ meets all requirements above.
The gadget is formally defined as follows.


\begin{definition}[$(U,k,\alpha,\varepsilon)$-cover family]
    A $(U,k,\alpha,\varepsilon)$-cover family  is a collection $\Scal\subseteq 2^{U}$ of subsets of $U$, such that
    \begin{itemize}
        \item (C1) For every $S\in\Scal$, $|S|\le\frac{(1+\varepsilon)\alpha |U|}{k}$.
        \item (C2) For every $\tilde{S}\subseteq U$ with $|\tilde{S}|\le\alpha |U|$, there exists $\{T_1,\dots,T_k\}\subseteq\Scal$ such that $\bigcup_{i\in[k]}T_i=\tilde{S}$. Moreover, $T_i\cap T_j=\emptyset$ for every $i\neq j\in[k]$.
    \end{itemize}
\end{definition}

Note that property (C1) immediately implies that for any $\tilde{S}\subseteq U$ with $|\tilde{S}|\ge \beta|U|$, it requires at least $\frac{\beta}{(1+\varepsilon)\alpha}k$ sets in $\Scal$ to cover $\tilde{S}$, hence producing our desired gap.
We now present the improved reduction as follows. 
\begin{lemma}\label{lem:cover_group}
    There exists a reduction that takes as input an integer $k\in\N$, an $([m],k,\alpha,\varepsilon)$-cover family $\Scal\subseteq 2^{[m]}$, and a $\gapmldq$ instance $(\vec M,\vec u,\alpha m)$ where $\vec M\in \mathbb{F}_q^{d\times m}$, $\vec{u} \in \mathbb{F}_q^d$. The reduction outputs a new matrix $\vec{M}_k\in\F_q^{d\times m'}$ in $(m')^{O(1)}$ time, which satisfies the following properties:
    \begin{itemize}
        \item \textbf{Size:} $m'= |\Scal|\cdot q^{\frac{(1+\varepsilon)\alpha m}{k}}$.
        
        \item \textbf{Completeness:}
        If there exists $\vec{x}\in\F_q^{m}$ with $\|\vec x\|_0\le\alpha m$ such that $\vec M\vec x=\vec{u}$, then there exists $\vec y\in\F_q^{m'}$ with $\|\vec y\|_0\le k$ such that $\vec M_k\vec y=\vec{u}$.
        
        \item \textbf{Soundness:}
        If there is no $\vec{x}\in\F_q^{m}$ with $\|\vec x\|_0\le\gamma\alpha m$ such that $\vec M\vec x=\vec{u}$, then there is no $\vec{y}\in\F_q^{m'}$ with $\|\vec y\|_0\le\frac{\gamma}{1+\varepsilon} k$ such that $\vec M_k\vec y=\vec{u}$.
    \end{itemize}
\end{lemma}
\begin{proof}
    Let $A:=\{\vec \alpha\in\F_q^m: \supp(\vec\alpha)\in\Scal\}$ and $m':=|A|$.
    We construct the new matrix $\vec M_k\in\F_q^{d\times m'}$ as follows.
    We establish a bijection between the column indices of $\vec M_k$ and $A$, and use $\vec\alpha\in A$ to denote a column index of $\vec M_k$.
    For each $\vec\alpha\in A$, we let the $\vec\alpha$-th column of $\vec M_k$ be
    $$
        \vec M_k[\vec\alpha] = \vec M\vec \alpha \in\F_q^d.
    $$
    The property (C1) of $\Scal$ guarantees that $|S|\le\frac{(1+\varepsilon)\alpha m}{k}$, hence the size of $\vec M_k$ satisfies
    $$
        m' = \sum_{S\in\Scal}(q-1)^{|S|} < |\Scal|\cdot q^{\frac{(1+\varepsilon)\alpha m}{k}}.
    $$
    
    \subparagraph{Completeness.}
    Assume that there exists $\vec{x}\in\F_q^{m}$ with $\|\vec x\|_0\le\alpha m$ such that $\vec M\vec x=\vec{u}$.
    Since $|\supp(\vec x)|\le\alpha m$, by the property (C2) of $\Scal$, there exists $T_1,\dots,T_k\in\Scal$ as a partition of $\supp(\vec x)$.

    For each $i\in[k]$, let $\vec x_i\in\F_q^m$ be the projection of $\vec{x}$ onto the coordinates in $T_i$, i.e., for every $j\in[m]$, we let
    $$
        \vec x_i[j]=\begin{cases}
            \vec x[j], & j\in T_i, \\
            0, & j\not\in T_i.
        \end{cases}
    $$
    Then $\sum_{i=1}^k\vec x_i=\vec x$, and $\vec x_i\in A$ because $\supp(\vec x)=T_i\in\Scal$.
    Let $\vec y\in\F_q^{m'}$ be such that $\vec y[\vec x_i]=1$ for every $i\in[k]$, and $\vec y[\vec\alpha]=0$ for any other $\vec \alpha\in A$.
    By construction, we have
    $$
        \vec M_k\vec y
        =\sum_{i=1}^k \vec M_k[\vec x_i]
        =\sum_{i=1}^k \vec M\vec x_i
        =\vec M\vec x=\vec u.
    $$
        
    \subparagraph{Soundness.}
    Assume for the sake of contradiction that there exists $\vec{y}\in\F_q^{m'}$ with $\|\vec y\|_0\le\frac{\gamma}{1+\varepsilon} k$ such that $\vec M_k\vec y=\vec{u}$.
    Let $X:=\supp(\vec y)$, then we have
    $$
        \vec u=\vec M_k\vec y=\sum_{\vec \alpha\in X}\vec M_k[\vec\alpha]\cdot \vec y[\vec\alpha] = \sum_{\vec\alpha\in X}\vec M\vec\alpha\cdot \vec y[\vec\alpha].
    $$
    Let $\vec x:=\sum_{\vec\alpha\in X}\vec y[\vec\alpha]\cdot \vec\alpha$, then $\vec M\vec x=\vec u$.
    By the property (C1) of $\Scal$, we have $|\supp(\vec\alpha)|\le\frac{(1+\varepsilon)\alpha}{k}m$ for every $\vec\alpha\in A$.
    Since $|X|=\|\vec y\|_0\le \frac{\gamma}{1+\varepsilon} k$,
    we have
    $$
        |\supp(\vec x)|\le\big|\bigcup_{\vec\alpha\in X}\supp(\vec\alpha)\big|\le |X|\cdot\frac{(1+\varepsilon)\alpha}{k}m \le \gamma\alpha m.
    $$
    Hence $\vec x$ satisfies $\vec M\vec x=\vec u$ and $\|\vec x\|_0\le\gamma\alpha m$, proving the soundness.
\end{proof}

Combining \cref{lem:maxlin2mld} and \cref{lem:cover_group}, we obtain a reduction which on input an $([m],k,1-c,\varepsilon)$-cover family $\Scal$, it reduces a $\gapmaxlin$ instance with $m$ equations to $\gapkmldq$ instance with size $N=|\Scal|\cdot 2^{O(m/k)}$ for $\gamma=\frac{1-s}{(1-c)(1+\varepsilon)}$.
If $|\Scal|=2^{O(m/k)}$, by \cref{cor:linEthHard} we get that assuming $\ETH $, there exists a constant $\gamma > 1$ such that no algorithm running in time $ N^{o(k)} $ can solve $\gapkmldq$.

To achieve the tight lower bound, we will construct a cover family over universe $[m]$ with size $|\Scal|= 2^{O(m/k)}$ using the following lemma, whose proof is postponed to  \cref{sec:constr}. 

\begin{lemma}
\label{lem:partition_exist}
    For every $m,k\in\N$ and $0<\alpha,\varepsilon<1$ such that $m\ge k^{4k^2/(\varepsilon^2\alpha)}$, there exists an $([m],k,\alpha,\varepsilon)$-cover family $\Scal$ with $|\Scal|=\lceil\log_k m\rceil\cdot k\cdot 2^{2m/k}$ that can be constructed deterministically in $|\Scal|^{O(1)}$ time.
\end{lemma}

It follows that
$$
    N =  2^{O_q(m/k)}.
$$
This improves upon the naive grouping bound of $2^{O(m \log k / k)}$ by removing the $\log k$ factor in the exponent. Consequently, under $\ETH$, this reduction implies that $\gapkmldq$ cannot be solved in time $N^{o(k)}$, confirming the tightness of the lower bound for the parameterized problem.

Note that \cref{lem:partition_exist} requires $m\ge k^{\Omega_{\alpha,\varepsilon}(k^2)}$.
To obtain the tight lower bound for larger $k$, we give a randomized construction of cover families.

\begin{lemma}
\label{lem:partition_random}
    For every $m,k\in\N$, $0<\alpha,\varepsilon<1$ such that $m \ge \frac{6k\ln 2k}{\varepsilon^2\alpha}$, there exists an $([m],k,\alpha,\varepsilon)$-cover family $\Scal$ with $|\Scal|=\frac{12k^2}{\varepsilon^2\alpha}\cdot 2^{(1+\varepsilon)m/k}$ that can be constructed in $|\Scal|^{O(1)}$ time with probability $1-o(1)$.
\end{lemma}

\subsection{Putting Everything Together}
We now combine everything together to prove the main result and derive other conclusions.

\begin{theorem}[\cref{thm: mld_main} Restated] \label{thm:mld_main_restate}
    Assuming $\ETH$, there exist constants $\gamma > 1$ and $\delta>0$, such that  for any algorithm that takes as input a matrix $\vec{H}\in \F_q^{d \times N}$, a vector $\vec{u}\in\F_q^{d}$ and a parameter $k\in\N$, it must take $ N^{\delta k}$ time to distinguish the following two cases:
    \begin{itemize}
        \item $\yes$ Case: There exists $\vec x\in\F_q^{N}$ with $\|\vec x\|_0\le k$ such that $\vec H\vec x=\vec u$.
        \item $\no$ Case: For every $\vec x\in\F_q^{N}$ with $\|\vec x\|_0\le \gamma k$, we have $\vec H\vec x\neq\vec u$.
    \end{itemize}
\end{theorem}
\begin{proof}
    We start with a $\gapmaxlinq$ instance with $n$ variables and $m\le Cn$ equations for some constant $0<s<c<1$ and $C>0$. By \Cref{cor:linEthHard}, there exists some constant $\delta>0$ such that deciding $\gapmaxlinq$ must take $2^{\delta n}$ time assuming $\ETH$.

    We now invoke \cref{lem:partition_exist} with $\alpha\leftarrow 1-c$, and a sufficiently small $\varepsilon>0$ such that $\gamma := \frac{1-s}{(1+\varepsilon)(1-c)} > 1$. This gives us a $\big([m],k,1-c,\varepsilon\big)$-cover family $\Scal$ with $|\Scal| = k\cdot 2^{O(m/k)}$, constructible in $|\Scal|^{O(1)}$ time.
    Using \cref{lem:maxlin2mld} and \cref{lem:cover_group}, we reduce this $\gapmaxlinq$ instance to a $\gapkmldq$ instance where the number of vectors is $N = |\Scal| \cdot q^{O(m/k)} = k\cdot  2^{O_q(n/k)}$.

    Assume that $\gapkmldq$ can be solved in $N^{\delta k}$ time for every $\delta>0$,
    then for any $\delta'>0$, there exists $\delta$ such that $\gapmaxlinq$ can be solved in time
    $\big(k\cdot 2^{O_q(n/k)}\big)^{\delta k}<2^{\delta' n}$.
    Therefore, there exist $\gamma > 1$ and $\delta>0$ such that assuming $\ETH$, any algorithm that decides $\gapkmldq$ must take time $N^{\delta k}$.
\end{proof}

In fact, \Cref{thm:mld_main_restate} also holds for larger $k$ such that $k=O((\log\log N)^{0.49})$, since we can construct the $\big([m],k,1-c,\varepsilon\big)$-cover family as long as $k^{\Omega(k^2)}\le m$. Following this observation, if we use the randomized construction of cover family instead of the deterministic one, we obtain the lower bound for any $k\le N^\varepsilon$ for some constant $\varepsilon>0$ under randomized $\ETH$.

\begin{theorem}[\cref{thm: mld_main_2} Restated] 
    Assuming randomized $\ETH$, there exist constants $\gamma > 1$ and $\delta,\varepsilon'>0$, such that for any algorithm that takes as input a matrix $\vec{H}\in \F_q^{d \times N}$, a vector $\vec{u}\in\F_q^{d}$ and an integer $k\in\N$ such that $2\le k\le N^{\varepsilon'}$, it must take $N^{\delta k}$ time to distinguish the following two cases:
    \begin{itemize}
        \item $\yes$ Case: There exists $\vec x\in\F_q^{d}$ with $\|\vec x\|_0\le k$ such that $\vec H\vec x=\vec u$.
        \item $\no$ Case: For every $\vec x\in\F_q^{d}$ with $\|\vec x\|_0\le \gamma k$, we have $\vec H\vec x\neq\vec u$.
    \end{itemize}
\end{theorem}
\begin{proof}
    We still start with a $\gapmaxlinq[1-\alpha]$ instance with $n$ variables and $m\le Cn$ equations. Pick a small enough constant $\varepsilon>0$ such that $\gamma:=\frac{1-s}{(1+\varepsilon)\alpha}>1$.
    Let
    $$
        N:=\frac{12k^2}{\varepsilon^2\alpha}\cdot 2^{(1+\varepsilon)(1+\alpha\log_2 q)\frac{m}{k}}
        \quad\text{and}\quad
        \beta:=\frac{\varepsilon^2\alpha}{6(1+\varepsilon)(1+\alpha\log_2 q)}.
    $$
    When
    $$
        k\le \Big(\frac{\varepsilon^2\alpha}{12}N\Big)^{\frac{\beta}{2(1+\beta)}},
    $$
    which implies that
    $$
        k\ln 2k\le k\log_2(2k)\le\frac{\varepsilon^2\alpha}{6}m,
    $$
    by \cref{lem:partition_random}, we can construct a $\big([m],k,\alpha,\varepsilon\big)$-cover family $\Scal$ with size $|\Scal|=\frac{12k^2}{\varepsilon^2\alpha}\cdot 2^{(1+\varepsilon)m/k}$.

    Using the $\big([m],k,\alpha,\varepsilon\big)$-cover family, we can reduce the $\maxlin$ instance to a $\gapkmldq$ instance with size $N=|\Scal| \cdot q^{(1+\varepsilon)\alpha m/k}=k^2\cdot 2^{O_q(n/k)}$.
    Assume that $\gapkmldq$ can be solved in $N^{\delta k}$ time for every $\delta>0$, then for any $\delta'>0$, there exists $\delta$ such that $\gapmaxlinq$ can be solved in time
    $$
        \big(k^2\cdot 2^{O_q(n/k)}\big)^{\delta k}
        = 2^{\delta\cdot O_q(n+k\log k)}
        <2^{\delta' n}.
    $$
    This proves the $N^{\delta k}$-time lower bound for $k=O(N^{\varepsilon'})$ where $\varepsilon':=\frac{\beta}{2(1+\beta)}$.
\end{proof}

We state the following corollary which can be obtained by a standard reduction from $\gapkmldq$ to $\gapkncpq$ as in \cite{LLL24}.
\begin{corollary}
    Assuming $\ETH$, there exists a constant $\gamma>1$, such that no algorithm can decide $\gapkncpq$ instance $(\vec A,\vec t)$ where $\vec A\in F_q^{N \times M}, \vec t \in \F_q^{N}$  in $N^{o(k)}$ time. 
\end{corollary}

Using the reduction from $\gapncp_2$ to $\gapcvpp$ as in \cite{agm25}, we can obtain optimal hardness for $\gapkcvpp$ for all $p>1$. 
\begin{corollary}
    Assuming $\ETH$, for every $p>1$, there exists a constant $\gamma>1$, such that no algorithm can decide $\gapkcvpp$ instance $(\vec A,\vec t)$ where $\vec A\in \Z^{N \times M}, \vec t \in \Z^{N}$  in $N^{o(k)}$ time.
\end{corollary}

%% file: construction.tex
\section{Construction of Cover Families}
\label{sec:constr}
In this section, we present the construction of cover families. Our approach relies on an intermediate combinatorial object which we call a \emph{balanced partition family}. A balanced partition family is a family of partitions over a universe of size $n$, such that every subset of size $\alpha n$ is partitioned into almost equal size by some partition. We first formalize balanced partition families as follows, and then establish a reduction showing that any efficient balanced partition family yields a corresponding cover family.

\begin{definition}[$(U,k,\alpha,\varepsilon,c)$-balanced partition family]
    A $(U,k,\alpha,\varepsilon,c)$-balanced partition family is a family $\Fcal\subseteq[k]^{U}$ of functions from $U$ to $[k]$, such that
    \begin{itemize}
        \item (P1) For every $f\in\Fcal$ and $j\in[k]$, $|f^{-1}(j)|\le \frac{c|U|}{k}$.
        \item (P2) For every $S\subseteq U$ with $|S|=\alpha |U|$, there exists $f\in\Fcal$ such that $|S\cap f^{-1}(j)|\le\frac{(1+\varepsilon)|S|}{k}$ for every $j\in[k]$.
    \end{itemize}
\end{definition}
In the following theorem we present the construction of a  \textit{cover family}, assuming the existence of a  \textit{balanced partition family}.
\begin{theorem}\label{thm:par2cov}
    On input a $(U,k,\alpha,\varepsilon,c)$-balanced partition family $\Fcal$, one can construct a $\big(U,k,\alpha,\varepsilon\big)$-cover family $\Scal$ with $|\Scal|=|\Fcal|\cdot k\cdot 2^{c|U|/k}$ in $|\Scal|^{O(1)}$ time.
\end{theorem}
\begin{proof}
    Let $m := |U|$. We construct the cover family $\Scal$ as follows:
    \begin{itemize}
        \item For every function $f \in \Fcal$ and every index $j \in [k]$, let $\Scal_{f,j}$ be the collection of all subsets of the $j$-th bucket $f^{-1}(j)$ that have size at most $\frac{(1+\varepsilon)\alpha m}{k}$. That is,
        $$
            \Scal_{f,j} := \left\{ T \subseteq f^{-1}(j) : |T| \le \frac{(1+\varepsilon)\alpha m}{k} \right\}.
        $$
        \item Define $\Scal := \bigcup_{f \in \Fcal} \bigcup_{j \in [k]} \Scal_{f,j}$.
    \end{itemize}

    \subparagraph{Size and Running Time.}
    By property (P1) of the balanced partition family, for every $f \in \Fcal$ and $j \in [k]$, the bucket size is bounded by $|f^{-1}(j)| \le \frac{cm}{k}$. The number of subsets in $\Scal_{f,j}$ is bounded by the total number of subsets of the bucket, which is $2^{|f^{-1}(j)|}$. Therefore,
    $$
        |\Scal| \le \sum_{f \in \Fcal} \sum_{j \in [k]} 2^{|f^{-1}(j)|} \le |\Fcal| \cdot k \cdot 2^{cm/k}.
    $$
    The construction can be performed in time polynomial in $|\Scal|$.

    \subparagraph{Property (C1).}
    By construction, the maximum size of any set $S \in \Scal$ is bounded by $\frac{(1+\varepsilon)\alpha m}{k}$.

    \subparagraph{Property (C2).}
    Let $\tilde{S} \subseteq U$ be any set with $|\tilde{S}| \le \alpha m$. We need to show that $\tilde{S}$ can be exactly covered by $k$ disjoint sets from $\Scal$.
    Pick an arbitrary set $S'$ such that $|S'| = \alpha m$ and $\tilde{S}\subseteq S'$.
    By property (P2) of the balanced partition family, there exists a function $f \in \Fcal$ such that for all $j \in [k]$,
    $$
        |\tilde{S} \cap f^{-1}(j)|\le |S' \cap f^{-1}(j)| \le \frac{(1+\varepsilon)|S'|}{k} = \frac{(1+\varepsilon)\alpha m}{k}.
    $$
    Define $T_j := \tilde{S} \cap f^{-1}(j)$ for each $j \in [k]$.
    \begin{itemize}
        \item Since $T_j \subseteq f^{-1}(j)$ and $|T_j| \le \frac{(1+\varepsilon)\alpha m}{k}$, we have $T_j \in \Scal_{f,j} \subseteq \Scal$.
        \item Since $\{f^{-1}(j)\}_{j \in [k]}$ forms a partition of $U$, the sets $T_1, \dots, T_k$ are pairwise disjoint.
        \item $T_1,\dots, T_k$ can cover $\tilde{S}$, because $\bigcup_{j \in [k]} T_j = \tilde{S} \cap \bigcup_{j \in [k]} f^{-1}(j) = \tilde{S} \cap U = \tilde{S}$.
    \end{itemize}
    Thus, $\{T_1, \dots, T_k\} \subseteq \Scal$ is a exact cover for $\tilde{S}$ and satisfies (C2).
\end{proof}

\subsection{Randomized Construction}
We will prove the existence of a balanced partition family with $O(k)$ size, by showing that
randomly sampling $O(k)$ partitions and then rejecting all partitions that do not satisfy (C1) produces a balanced partition family with high probability.

\begin{theorem}\label{thm:rand_par}
    For every $m,k\in\N$ and $0<\alpha,\varepsilon<1$ such that $m\ge\frac{6k\ln 2k}{\varepsilon^2\alpha}$,
    there exists an $([m],k,\alpha,\varepsilon,1+\varepsilon)$-balanced partition family $\Fcal$ with $|\Fcal|=\frac{12k}{\varepsilon^2\alpha}$.
\end{theorem}
\begin{proof}
    Let $t := \big\lceil \frac{12k}{\varepsilon^2\alpha} \big\rceil$. We sample $t$ functions $f_1, \dots, f_t: U \to [k]$ independently and uniformly at random.
    We define $I\subseteq[k]$ as the indices of functions that satisfy (P1), formally, let
    $$
        I := \left\{ i \in [t] : \forall j \in [k], \ |f_i^{-1}(j)| \le \frac{(1+\varepsilon)m}{k} \right\}.
    $$
    For any subset $S \subseteq U$, let $B(S)$ be the ``bad event'' that no function $f_i$ with $i\in I$ evenly splits $S$.
    That is, for every $i \in I$, there exists some bucket $j \in [k]$ such that $|S \cap f_i^{-1}(j)| > \frac{(1+\varepsilon)|S|}{k}$.
    We will show that with high probability, $B(S)$ does not occur for every subset $S \subseteq [m]$ of size $\alpha m$.

    Consider a fixed index $i \in [t]$ and a fixed subset ${S} \subseteq U$ with $|{S}| = \alpha m$. We define two bad events for $f_i$ and analyze the probability of these events:
    \begin{itemize}
        \item Let $E_1(i)$ be the event that $i \notin I$.
        Since $f_i$ is a uniform random function, $f(x)$ is drawn uniformly and independently from $[k]$ for each $x\in[m]$. Hence for a fixed bucket $j \in [k]$,  $x\in f^{-1}_i(j)$ with probability $\frac{1}{k}$ for every $x\in[m]$, and thus $|f_i^{-1}(j)|$ is a sum of independent Bernoulli trials with mean $m/k$. By the Chernoff bound,
        $$
            \Pr\left[ |f_i^{-1}(j)| > \frac{(1+\varepsilon)m}{k} \right] < \exp\left( -\frac{\varepsilon^2 m}{3k} \right).
        $$
        Taking a union bound over all $k$ buckets,
        $$
            \Pr[E_1(i)] < k \exp\left( -\frac{\varepsilon^2 m}{3k} \right).
        $$

        \item Let $E_2(i, {S})$ be the event that $f_i$ fails the  property (P2) for the specific set ${S}$. This means there exists some bucket $j \in [k]$ such that $|{S} \cap f_i^{-1}(j)| > \frac{(1+\varepsilon)\alpha m}{k}$. For a fixed $j \in [k]$, the random variable $|{S} \cap f_i^{-1}(j)|$ follows a binomial distribution $\Bin(\alpha m, 1/k)$. Similarly, by the Chernoff bound and a union bound over $k$ buckets, we have
        $$
            \Pr[E_2(i, {S})] < k \exp\left( -\frac{\varepsilon^2 \alpha m}{3k} \right).
        $$
    \end{itemize}

    Let $E(i, {S}) := E_1(i) \lor E_2(i, {S})$ be the combined bad event where $f_i$ either violates (P1) or fails to balance ${S}$. Since $\alpha < 1$, we have
    $$
        \Pr[E(i, {S})] \le \Pr[E_1(i)] + \Pr[E_2(i, {S})] < 2k \exp\left( -\frac{\varepsilon^2 \alpha m}{3k} \right).
    $$
    Since the functions are sampled independently, we have
    $$
        \Pr[B(S)] = \prod_{i=1}^t \Pr[E(i, {S})] < \bigg(2k\exp\Big({ -\frac{\varepsilon^2\alpha m}{3k} }\Big)\bigg)^t.
    $$
    Take a union bound over all possible subsets ${S}$ of size $\alpha m$. Since the total number of such subsets is $\binom{m}{\alpha m} <\e^m$, the probability that there exists a bad set ${S}$ is at most
    $$\begin{aligned}
        \sum_{S\in\binom{[m]}{\alpha m}} \Pr[B(S)]
        &< \e^m\bigg(2k\exp\Big({ -\frac{\varepsilon^2\alpha m}{3k} }\Big)\bigg)^t \\
        &= \exp\bigg({ m -t\Big(\frac{\varepsilon^2\alpha m}{3k} -\ln 2k\Big) }\bigg).
    \end{aligned}$$
    When $m\ge\frac{6k\ln 2k}{\varepsilon^2\alpha}$ and $t=\big\lceil \frac{12k}{\varepsilon^2\alpha} \big\rceil$, the total probability is at most $\e^{-m}$.
    
    We let the final family $\Fcal$ contain all functions that satisfy property (P1), i.e., let $\Fcal := \{f_i : i \in I\}$.
    By definition, every function in $\Fcal$ satisfies (P1).
    Furthermore, we have proved that with probability $1-\e^{-m}$, for every subset $S\in\binom{U}{\alpha m}$, there exists $f\in\Fcal$ that evenly splits $S$, hence $\Fcal$ also satisfies (P2).
    Thus, $\Fcal$ is a valid balanced partition family with high probability.
\end{proof}

Using \cref{thm:par2cov}, we obtain an $([m],k,\alpha,\varepsilon)$-cover family $\Scal$ with $|\Scal|=\frac{12k}{\varepsilon^2\alpha}\cdot k\cdot 2^{(1+\varepsilon)m/k}$ in $|\Scal|^{O(1)}$ time, hence proving \cref{lem:partition_random}.


\subsection{Derandomization}
This section introduces the hypercube partition system as a derandomization of the balanced partition family. We define the hypercube partition system as follows.
\begin{definition}[Hypercube Partition System]
    A $(k,d)$-hypercube partition system is a pair $(U,\Fcal)$ where $U:=[k]^{[d]}$ and $\Fcal:=\{f_1,\dots,f_d\}$ is a collection of funtions from $U$ to $[k]$. We view $U$ as the set of all points in a $d$-dimensional hypercube where each coordinate takes a value in $[k]$. The $i$-th partition $f_i:U\to[k]$ is defined as for every $x\in U$
    $$
        f_i(x)=x(i).
    $$
    In other words, for any $x\in U$, $f_i$ groups $x$ according to its $i$-th coordinate, into the $x(i)$-th set.
\end{definition}

\input{lemma4.6_new}

%% file: lemma4.6_new.tex

We show that a $(k,d)$-hypercube partition system $(U,\Fcal)$ is a balanced partition family for sufficiently large $d$. Informally, it suffices to show the following proposition:
\begin{itemize}
    \item[($\star$)] For any subset $S\subseteq U$, if on every coordinate $i\in[d]$, there exists some $j \in [k]$ such that at least $\frac{(1+\varepsilon)|S|}{k}$ elements $x \in S$ satisfy $x(i) = j$, then $|S|<\alpha |U|$. 
\end{itemize}
By the contrapositive of ($\star$), if $S$ is large enough, then at least one coordinate $i\in[d]$ must evenly split $S$. This indicates that $\Fcal$ is a balanced partition family.



To show $(\star)$, we consider a uniform distribution $p$ over the subset $S$.
Let $p_i$ denote the marginal distribution of $p$ on the $i$-th coordinate.
The condition in $(\star)$ guarantees that for each coordinate $i$, $p_i$ concentrates on some value $j$, which implies that $p_i$ has small entropy.
Then we add up the entropy of $p_i$ over all coordinates, which provides a upper bound on the entropy of $p$. Hence $|S|$ must be small.

We prove this formally in \cref{lem:partition_gadget_derandomize_lem}.

\begin{lemma}
\label{lem:partition_gadget_derandomize_lem}
    Let $d,k\in\N$, $0<\alpha,\varepsilon<1$, and $U:=[k]^{d}$.
    For every $i\in[d]$, let $f_i:U\to[k]$ be such that $f_i(x)=x_i$.
    If $d\ge\frac{4k}{\varepsilon^2\alpha}$, then for every $S\subseteq U$ with $|S|=\alpha k^d$, there exists $i\in[d]$, such that $|S\cap f_i^{-1}(j)|\le (1+\varepsilon)\alpha k^{d-1}$ for every $j\in[k]$.
\end{lemma}

\begin{proof}
    Let $X$ be a random variable uniformly distributed over the subset $S$. We have $H(X)=\ln |S|=d\ln k+\ln\alpha$.
    Let $p_i$ be the distribution of $X_i$, i.e., for every $j\in[k]$,
    $$
        p_i(j) = \frac{|\{x\in S:x_i=j\}|}{|S|}=\frac{|S\cap f_i^{-1}(j)|}{\alpha k^d}.
    $$
    Now we will prove that there exists some $i\in[d]$, such that $p_i(j)\le \frac{1+\varepsilon}{k}$ for every $j\in[k]$, which immediately implies that $|S\cap f_i^{-1}(j)|\le(1+\varepsilon)\alpha k^{d-1}$.

    Assume for the sake of contradiction that for every $i\in[d]$, there exists some $j_i\in[k]$ such that $p_i(j_i)>\frac{1+\varepsilon}{k}$.
    Fix any $i\in[d]$ and let $\rho_i:=p_i(j_i)$.
    Then, by Jensen's inequality,
    $$\begin{aligned}
        H(X_i) &= p_i(j_i)\ln\frac{1}{p_i(j_i)}+\sum_{j\neq j_i}p_i(j)\ln\frac{1}{p_i(j)} \\
        &\le \rho_i\ln\frac{1}{\rho_i}+(1-\rho_i)\ln\frac{k-1}{1-\rho_i}.
    \end{aligned}$$
    Let $h(x):=x\ln\frac{1}{x}+(1-x)\ln\frac{k-1}{1-x}$.
    We have $h'\big(\frac{1}{k}\big)=0$ and $h''(x)=-\frac{1}{x(1-x)}$.
    By Taylor's theorem with the Lagrange's form of remainder, we expand $h(x)$ at $x=1/k$, then there exists $\xi$ with $\frac{1}{k}<\xi< \rho_i$ such that
    $$\begin{aligned}
        h(\rho_i) &= h\Big(\frac{1}{k}\Big)+h'\Big(\frac{1}{k}\Big)\Big(\rho_i-\frac{1}{k}\Big)+\frac{h''(\xi)}{2}\Big(\rho_i-\frac{1}{k}\Big)^2 \\
        &= \ln k - \frac{1}{2\xi(1-\xi)}\Big(\rho_i-\frac{1}{k}\Big)^2.
    \end{aligned}$$
    Recall that $\rho_i>\frac{1+\varepsilon}{k}$. Since $\xi(1-\xi)<\xi<\rho_i$, and by the fact that $f(x):=x+\frac{1}{x}-2$ is increasing when $x>1$, we have
    $$\begin{aligned}
        \ln k-h(\rho_i) > \frac{1}{2\rho_i}\Big(\rho_i-\frac{1}{k}\Big)^2
        = \frac{f(k\rho_i)}{2k} > \frac{f(1+\varepsilon)}{2k}
        > \frac{\varepsilon^2}{4k},
    \end{aligned}$$
    thus
    $$
        H(X_i)\le h(\rho_i)<\ln k-\frac{\varepsilon^2}{4k}.
    $$
    By the subadditivity of entropy, when $d\ge \frac{4k}{\varepsilon^2\alpha}$, we have
    $$
        H(X)\le \sum_{i\in[d]} H(X_i)\le d\ln k-\frac{\varepsilon^2d}{4k}
        \le d\ln k-\frac{1}{\alpha}
        <d\ln k+\ln\alpha,
    $$
    which contradicts $H(X)=\ln|S|=d\ln k+\ln\alpha$.
\end{proof}

\cref{lem:partition_gadget_derandomize_lem} implies that a $(k,d)$-hypercube partition system is a $([k]^{d},k,\alpha,\varepsilon,1)$-balanced partition family of size $d$ if $d$ is large enough.
A limitation of this construction is that we require the size of the universe to be an integer power of $k$.
To address this, we define a subset $U \subseteq [k]^{d}$ of size exactly $m$ by taking a ``diagonal slice'' of the hypercube. This slice preserves the balanced property of projections: fixing one coordinate to any value leaves the remaining coordinates to cycle through values uniformly, ensuring the slice is spread evenly across buckets.

\begin{theorem}\label{thm:partition_gadget_derandomize}
    For every $m,k\in\N$ and $0<\eta,\varepsilon<1$ such that $m\ge k^{4k^2/(\varepsilon^2\eta)}$, there exists an $([m],k,\eta,\varepsilon,2)$-balanced partition family $\Fcal$ of size $|\Fcal|=\lceil\log_k m\rceil$ that can be constructed in $m^{O(1)}$ time.
\end{theorem}

\begin{proof}
    Let $d := \lceil \log_k m \rceil$ and $c := \big\lceil\frac{m}{k^{d-1}}\big\rceil$.
    We have $(c-1)k^{d-1}<m\le ck^{d-1}$ and $2\le c\le k$.
    Define a partition $U_0,\dots,U_{k-1}$ of $[k]^d$ based on the modular sum of coordinates:
    $$
        U_\ell:=\left\{ x \in [k]^d : \left( \sum_{j=1}^d x_j \right) \bmod k =\ell \right\}.
    $$
    Let $U':=U_0\cup\dots\cup U_{c-1}$.
    Then $|U'|=ck^{d-1}$ because $|U_\ell|=k^{d-1}$ for every $\ell$.
    We set the universe $U$ to be an arbitrary subset of $U'$ of size $m\le |U'|$.
    Define $\Fcal = \{f_1, \dots, f_d\}$ as the projections on each coordinate: $f_i(x) = x(i)$.
    
    \subparagraph{Property (P1).}
    For any $i\in[d],j\in[k]$ and $0\le\ell<k$, the number of $x\in U_\ell$ such that $x_i=j$ is exactly $k^{d-2}$.
    To prove this, we take an arbitrary $i'\neq i$. If we fix $x_t$ for every $t\in[d]\setminus\{i,i'\}$, then $x_{i'}$ must be $x_{i'}=\ell-\sum_{t\in[d]\setminus\{i'\}}x_t$. The number of ways to fix $x_t$ is $k^{d-2}$, hence there are $k^{d-2}$ many $x\in U_\ell$ with $x_i=j$, i.e.,
    $$
        |U_\ell\cap f_i^{-1}(j)|=k^{d-2}.
    $$
    Since $m>(c-1)k^{d-1}$, for any $i\in[d]$ and $j\in[k]$,
    $$
        |U \cap f_i^{-1}(j)| \le |U' \cap f_i^{-1}(j)| = c k^{d-2} \le 2(c-1)k^{d-2} < \frac{2m}{k}.
    $$
    Hence $\Fcal$ satisfies property (P1).
    
    \subparagraph{Property (P2).}
    Let $\alpha := \frac{\eta m}{k^d}$. Since $m > k^{d-1}$, we have $\eta < \alpha k$, hence
    $$
        d \ge \log_k m
        \ge \frac{4k^2}{\varepsilon^2\eta}
        > \frac{4k}{\varepsilon^2\alpha}.
    $$
    Applying \cref{lem:partition_gadget_derandomize_lem},
    for any $S \subseteq [k]^{d}$ of size $\alpha k^d=\eta m$, there exists $f_i\in\Fcal$ such that every bucket $S\cap f_i^{-1}(j)$ has size at most
    $$
        (1 + \varepsilon) \alpha k^{d-1} = \frac{(1+\varepsilon)\eta m}{k} = \frac{(1+\varepsilon)|S|}{k}.
    $$
    Hence $\Fcal$ also satisfies property (P2).
\end{proof}

Using \cref{thm:par2cov}, we obtain an $([m],k,\eta,\varepsilon)$-cover family $\Scal$ with $|\Scal|=\lceil\log_k m\rceil\cdot k\cdot 2^{2m/k}$ in $|\Scal|^{O(1)}$ time, hence proving \cref{lem:partition_exist}.